%% file: v4.tex
\documentclass{llncs}
\usepackage{llncsdoc}
\usepackage{epsfig}
\usepackage{color}

\newcommand{\Q}{\rm \bf \cal{Q}}

\newcommand{\region}{{\cal R}}
\newcommand{\inD}{I}
\newcommand{\outD}{O}
\newcommand{\opt}{OPT}

\newcommand{\recSep}{\mbox{recSep}}

\begin{document}

\title{\Large On Isolating Points using Disks}

\author{{Matt Gibson}\inst{1} \and
{Gaurav Kanade}\inst{2} \and {Kasturi Varadarajan}\inst{2}}
\institute{Department of Electrical and Computer Engineering\\
The University of Iowa\\
Iowa City, IA 52246\\
Email:matthew-gibson@uiowa.edu
\and
Department of Computer Science\\
The University of Iowa\\
Iowa City, IA 52246\\
Email:gaurav-kanade@uiowa.edu, kvaradar@iowa.uiowa.edu
}
\maketitle

\begin{abstract}
In this paper, we consider the problem of choosing disks (that we can think of as corresponding to wireless sensors) so that given a set of input points in the plane, there exists no path between any
pair of these points that is not intercepted by some disk. We try to achieve this separation using a minimum number of a given set of unit disks. We show that a constant factor approximation
to this problem can be found in polynomial time using a greedy algorithm. To the best of our knowledge we are the first to study this optimization problem.
\end{abstract}

\section{Introduction}
Wireless sensors are being extensively used in applications to provide barriers as a defense mechanism against intruders at important buildings, estates, national borders etc. Monitoring the area of interest by this type of coverage is called \textit{barrier} coverage \cite{KumarLA05}. Such sensors are also being used to detect and track moving objects such as animals in national parks, enemies in a battlefield, forest fires, crop diseases
etc. In such applications it might be prohibitively expensive to attain blanket coverage but sufficient to ensure that the object under consideration cannot travel too far before it is detected. Such a coverage is called \textit{trap} coverage \cite{BalisterZKS09,SankararamanERT09}.

Inspired by such applications, we consider the problem of isolating a set of points by a minimum-size subset of a given subset of unit radius disks. A unit disk crudely models the region sensed by a sensor, and the 
work reported here readily generalizes to disks of arbitrary, different radii.   
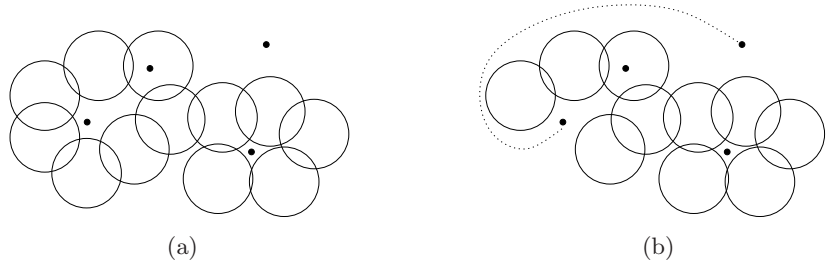
\begin{figure}[h]
\centering
\begin{tabular}{c@{\hspace{0.1\linewidth}}c}
\input{separate.pstex_t} & 
\input{separate2.pstex_t} \\ 
(a) & (b)
\end{tabular}
\caption{(a) This set of disks separates the points because every path
connecting any two points must intersect a disk.
(b) This set of disks does not separate the points.}
\label{fig:separate}
\end{figure}
\paragraph{Problem Formulation.}
The input to our problem is a set $\inD$ of $n$ unit disks, and a set $P$ of $k$ points such that $\inD$ 
{\em separates} $P$, that is, for any two points $p,q \in P$, every path between $p$ and $q$ intersects at least one disk in $\inD$. The goal is to find a minimum cardinality subset of $\inD$ that separates $P$. See Figure \ref{fig:separate} for an illustration of this notion of separation.

There has been a lot of recent interest on geometric variants of well-known NP-hard combinatorial optimization problems, and our work should be seen in this context. For several variants of the geometric set cover problem, for example, approximation algorithms have been designed \cite{ClarksonV05,AronovES10,MustafaR09} that improve upon the best guarantees
for the combinatorial set cover problem. For the problem of covering points by the smallest subset of a 
given set of unit disks, we have approximation algorithms that guarantee an $O(1)$ approximation and even a
PTAS \cite{bg95,MustafaR09}. These results hold even for disks of arbitrary radii. Our problem can be viewed as a set cover problem where the elements that need to be covered are not points, but paths. However, known results only 
imply a trivial $O(n)$ approximation when viewed through this set cover lens. 

Another example of a problem that has received such attention is the {\em independent set} problem. For many geometric variants \cite{ChalermsookC09,ChanH09,FP}, approximation ratios that are better than that for the combinatorial case are known.

Our problem is similar to the node multi-terminal cut problem in graphs \cite{GargVY04}. Here, we are given a graph
$G = (V,E)$ with costs on the vertices and a subset $U \subseteq V$ of $k$ vertices, and our goal is to 
compute a minimum cost subset of vertices whose removal disconnects every pair of vertices in $U$. This problem admits a poly-time algorithm that guarantees an $O(1)$ approximation. We note however that the problem we
consider does not seem to be a special case of the multi-terminal cut problem.             

\paragraph{Contribution and Related Work.}
Our main result is a polynomial time algorithm that guarantees an $O(1)$ approximation for the problem. To the best of our knowledge, this is the first non-trivial approximation algorithm for this problem. Our algorithm is simple and combinatorial and is in fact a greedy algorithm. We first present an $O(1)$ approximation for the following two-point separation problem. We are given a set of unit disks $G$, and 
two points $s$ and $t$, and we wish to find the smallest  subset $B \subseteq G$ so that $B$ separates    $s$ and $t$. 

Our greedy algorithm to the overall problem applies the two-point separation algorithm to find the 
cheapest subset $B$ of $I$ that separates some pair of points in $P$. Suppose that $P$ is partitioned
into sets $P_1,P_2,\ldots,P_{\tau}$ where each $P_i$ is the subset of points in the same ``face'' with 
respect to $B$. The algorithm then recursively finds a separator for each of the $P_i$, and returns
the union of these and $B$. 

The analysis to show that this algorithm has the $O(1)$ approximation guarantee relies on the combinatorial complexity of the boundary of the union of disks. It uses a subtle and global argument to bound
the total size of all the separators $B$ computed in each of the recursive calls.\footnote{In the
earlier version of this paper, a similar algorithm was analyzed in a more ``local'' fashion. The
basic observation was that the very first separator $B$ that is computed has size $O(|\opt|/k)$,
where $\opt$ is the optimal solution for the problem. Subsequent separators computed in the recursive
calls may be more expensive, but it was shown that the overall size is $O((\log k)\cdot |\opt|)$. In contrast,
the present analysis does not try to bound the size of the individual separators, but just the sum
of their sizes. As a consequence, the analysis also turns out to be technically simpler.}

Our approximation algorithm for the two-point separation problem, which is a subroutine we use in the overall algorithm, is similar to fast algorithms for
finding minimum $s$-$t$ cuts in undirected planar graphs, see for example \cite{Reif83}.  Our overall greedy algorithm has some resemblance to the algorithm of Erickson and Har-Peled \cite{EricksonH04} employed in the context of approximating the minimum cut graph of a polyhedral manifold. The details of the our algorithm and the analysis, however, are quite different from these papers since we do not have an embedded graph but rather a system of unit disks.
Sankararaman et al.
\cite{SankararamanERT09} investigate a notion of coverage which they call
{\em weak coverage}. Given a region
$\region$ of interest (which they take to be a square in the plane) and a set
$I$ of unit disks (sensors), the region is said to be $k$-weakly covered if each connected
component of $\region - \bigcup_{d \in I} d$ has diameter at most $k$. 
They consider the situation when a
given set  $I$ of unit disks {\em completely} covers $\region$, and address the problem of partitioning $I$ into  as many subsets as possible so that $\region$ is
$k$-weakly covered by every subset. Their work differs in flavor from ours mainly due to the assumption that $I$ completely covers $\region$.


\paragraph{Organization.} In Section \ref{sec:prelims}, we discuss standard notions we require, and then 
reduce our problem to the case where none of the points in input $P$ are contained in any of the input disks.
In Section \ref{sec:two-point}, we present our approximation algorithm for separating two points. 
In Section \ref{sec:improved}, we
describe our main result, the constant factor approximation algorithm for separating $P$.
We conclude in Section \ref{sec:conclusions} with some remarks.

\section{Preliminaries}
\label{sec:prelims}
We will refer to the standard notions of vertices, edges, and faces in arrangements of circles \cite{AgarwalS98}. In 
particular, for a set $R$ of $m$ disks, we are interested in the faces in the complement of the union of
the disks in $R$. These are the connected components of the set $\Re^2 - \bigcup_{d \in R} d$. We also
need the combinatorial result that the number of these faces is $O(m)$. Furthermore, the total 
number of vertices and edges on the boundaries of all these faces, that is, the combinatorial complexity of the boundary of the union of disks in $R$, is $O(m)$ \cite{AgarwalS98}. We make standard general position assumptions 
about the input set $\inD$ of disks in this article. This helps simplify the exposition and is without 
loss of generality.

\begin{lemma}
\label{claim:many}
Let $R$ be a set of disks in the plane, and $Q$ a set of points so that (a) no point from $Q$ is contained in any disk from $R$, and (b) no face in the complement of the union of the disks in $R$ contains more than 
one point of $Q$. Then $|R| = \Omega(|Q|)$.
\end{lemma}

\begin{proof}
The number of faces in the in the complement of the union of the disks in $R$ is $O(|R|)$.
\qed
\end{proof}

\paragraph{Covering vs. Separating.}
The input to our problem is a set $\inD$ of $n$ unit disks, and $P$ a set of $k$ points such that $\inD$ separates $P$. Let $P_c \subseteq P$ denote those points contained in some disk in $I$; and $P_s$ denote the 
remaining points. We compute an $\alpha$-approximation to the smallest subset of $\inD$ that covers 
$P_c$ using a traditional set-cover algorithm; there are several poly-time algorithms that guarantee that
$\alpha = O(1)$. We compute a $\beta$-approximation to the smallest subset of $\inD$ that separates
$P_s$, using the algorithm developed in the rest of this article. We argue below that the combination 
of the two solutions is an $O(\alpha + \beta)$ approximation to the smallest subset of $\inD$ that separates
$P$.

Let $\opt \subseteq \inD$ denote an optimal subset that separates $P$. Suppose that $\opt$ covers
$k_1$ of the points in $P_c$ and let $k_2 = |P_c| - k_1$. By Lemma \ref{claim:many}, $|\opt| =
\Omega(k_2)$.

Now, by picking one disk to cover each of the $k_2$ points of $P_c$ not covered by $\opt$, we see 
that there is a cover of $P_c$ of size at most $|\opt| + k_2 = O(|\opt|)$. Thus, our $\alpha$-approximation
has size $O(\alpha) \cdot |\opt|$. Since $\opt$ also separates $P_s$, our $\beta$-approximation has size 
$O(\beta) \cdot |\opt|$. Thus the combined solution has size $O(\alpha + \beta) \cdot |\opt|$.

In the rest of the article, we abuse notation and assume that no point in the input set $P$ is contained
in any disk in $\inD$, and describe a poly-time algorithm that computes an  $O(1)$-approximation to 
the optimal subset of $\inD$ that separates $P$. 

\section{Separating Two Points}
\label{sec:two-point}

Let $s$ and $t$ be two points in the plane, and $G$  a set of
disks such that no disk in $G$ contains either $s$ or $t$, but $G$ separates $s$ and $t$. See Figure \ref{fig:fig1}.
Our goal is to find the smallest cardinality subset $B$ of $G$ that separates
$s$ and $t$. We describe below a polynomial time algorithm that returns a constant factor
approximation to this problem.

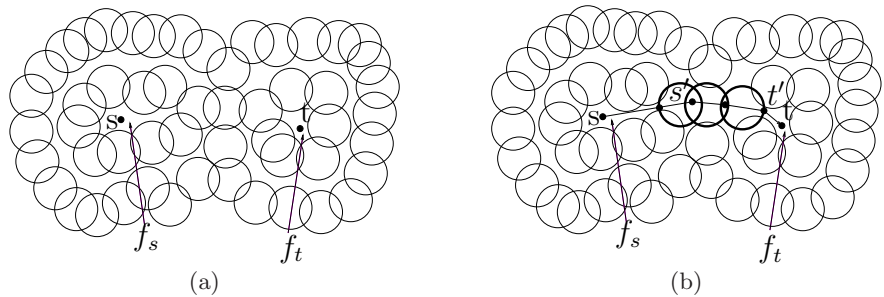
\begin{figure}[h]
\centering
\begin{tabular}{c@{\hspace{0.1\linewidth}}c}
\input{fig111.pstex_t} & 
\input{fig112.pstex_t} \\ 
(a) & (b)
\end{tabular}
\caption{(a) The figure shows faces $f_s$ and $f_t$
(b) This figure shows the sequence of disks in $\sigma$ (their boundaries are bold) and the path $\pi$.}
\label{fig:fig1}
\end{figure}


Without loss of generality, we may assume that the intersection graph of $G$ is connected. 
(Otherwise, we apply the algorithm to each connected component for which the disks in the component
separate $s$ and $t$. We return the best solution obtained.) Let $f_s$ and $f_t$ denote the faces
containing $s$ and $t$, respectively, in the arrangement of $G$. We augment the intersection
graph of $G$ with vertices corresponding to $s$ and $t$, and add an edge from $s$ to each disk
that contributes an edge to the boundary of the face $f_s$, and an edge from $t$ to each disk that 
contributes an edge to the boundary of the face $f_t$. We assign a cost of $0$ to $s$, $t$, and a cost of $1$ to each disk in $G$. We then find the shortest path from $s$ to 
$t$ in this graph, where the length of a path is the number of the
{\em vertices} on it that correspond to disks in $G$. Let $\sigma$ denote the sequence of disks on this shortest path. Note that any two disks that are not consecutive in
$\sigma$ do not intersect.

Using $\sigma$, we compute a path $\pi$ in the plane, as described below, from $s$ to $t$ so that (a) there are points $s'$ and
$t'$ on $\pi$ so that the portion of $\pi$ from $s$ to $s'$ is in $f_s$, and the portion from $t'$ to
$t$ is in $f_t$; (b) every point on $\pi$ from $s'$ to $t'$ is contained in some disk from $\sigma$; (c) the intersection of $\pi$ with each disk in $\sigma$ is connected. See Figure \ref{fig:fig1}.

Suppose that the sequence of disks in $\sigma$ is $d_1, \ldots, d_{|\sigma|}$. Let $s'$ (resp. $t'$)
be a point in $d_1$ (resp. $d_{\sigma}$) that lies on the boundary of $f_s$ (resp. $f_t$). For $1
\leq i \leq |\sigma| - 1$, choose $x_i$ to denote an arbitrary point in the intersection of $d_i$
and $d_{i+1}$. The path $\pi$ is constructed as follows: Take an arbitrary path from $s$ to $s'$ that
lies within $f_s$, followed by the  line segments $\overline{s' x_1}, \overline{x_1 x_2}, \ldots,
\overline{x_{|\sigma| -2} x_{|\sigma| - 1}}, \overline{x_{|\sigma|-1}t'}$, followed by an arbitrary path from
$t'$ to $t$ that lies within $f_t$.

Properties (a) and (b) hold for $\pi$ by construction. Property (c) is seen to follow from the fact that disks that are not consecutive in $\sigma$ do not overlap.
   
Notice that $\pi$ ``cuts'' each disk in $\sigma$ into two pieces. (Formally, the removal of $\pi$ from
any disk in $\sigma$ yields two connected sets.) The path $\pi$ may also intersect
other disks and cut them into two or more pieces, and we refer to these pieces as disk pieces. For a
disk that $\pi$ does not intersect, there is only one disk piece, which is the disk itself. 


We consider the intersection graph $H$ of the  disk pieces that come from disks in $G$. Observe that 
a disk piece does not have points on $\pi$, since $\pi$ is removed; so two disk pieces
intersecting  means there is a point outside $\pi$ that lies in both of them. In this graph, each disk piece has a cost of $1$.

 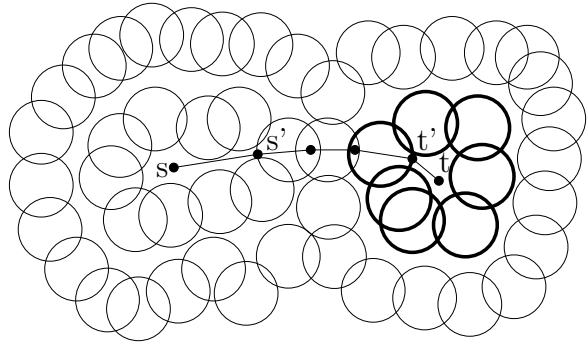
\begin{figure}[h]
 \centering
 \input{fig11.pstex_t} 
 \caption{This figure continues with the example of Figure \ref{fig:fig1}. The disks with bold boundary are the set $D$ computed by our algorithm. The only disk from $G$ with bold boundary has two disk pieces, and the shortest path between them in graph $H$ yields $D$.}
 \label{fig:fig2}
 \end{figure}

In this graph $H$, we compute, for each disk $d \in \sigma$, the shortest path between the two pieces 
corresponding to $d$. Suppose $d' \in \sigma$ yields the overall shortest path $\sigma'$; let $D$ denote 
the set of disks that contribute a disk piece to this shortest path. Our algorithm returns $D$
as its computed solution. See Figure \ref{fig:fig2}.

We note that $D$ separates $s$ and $t$ -- in particular, the union of the disk pieces in $\sigma'$ and the
set $\pi \cap d'$ contains a cycle in the plane that intersects the path $\pi$ between $s$ and $t$ 
exactly once.

\subsection{Bounding the Size of the Output}

Let $B^*$ denote the smallest subset of $G$ that separates $s$ and $t$. We will show that
$|D| = O(|B^*|)$. Let $f^*$ denote the face containing $s$ in the arrangement of $B^*$. Due to the optimality of $B^*$, we may
 assume that the boundary of $f^*$ has only one component. Let $a$ (resp. $b$) denote the first (resp. last) point on path $\pi$ where $\pi$ leaves $f^*$. It is possible that $a = b$. We find a minimum cardinality contiguous subsequence $\overline{\sigma}$ of $\sigma$ that contains the subpath of $\pi$ from $a$ to $b$; let $d_a$ and $d_b$ denote the first and last
disks in $\overline{\sigma}$. See Figure \ref{fig:fig4}.

\begin{figure}[h]
\centering
\begin{tabular}{c@{\hspace{0.1\linewidth}}c}
\input{fignew42.pstex_t} & 
\input{fignew421.pstex_t} \\ 
(a) & (b)
\end{tabular}
\caption{The shaded disks are in $B^*$.(a) This figure shows points $a$ and $b$ where $\pi$ leaves $f^*$ for the first and last time, respectively. 
(b) This figure shows the face $f$ containing $s$ in the arrangement with $ B^* \cup \overline{\sigma} $}
\label{fig:fig4}
\end{figure}
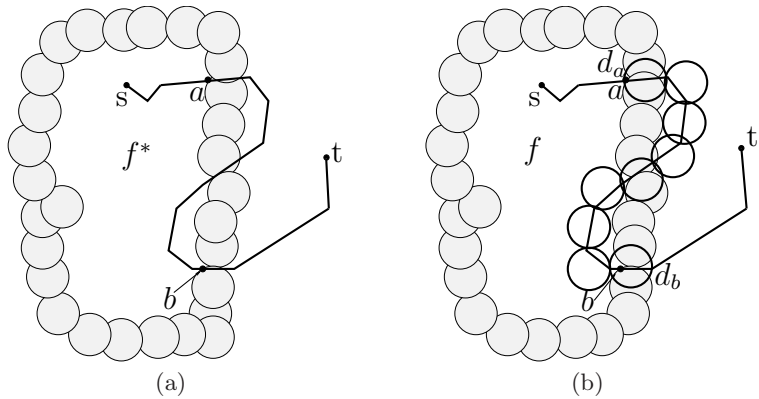

We claim that $|\overline{\sigma}| \leq |B^*| + 2$; if this inequality does not hold, then we
obtain a contradiction to the optimality of $\sigma$ by replacing the disks in the $\overline{\sigma} \setminus \{d_a, d_b\}$ by $B^*$.

Consider the face $f$ containing $s$ in the arrangement with $B^* \cup \overline{\sigma}$.
Each edge that bounds this face comes from a single disk piece, except for one edge corresponding to $d_a$ that may come from two disk pieces. (This follows from the fact that the portion of $\pi$ between $a$ and $b$ 
is covered by the disks in $\overline{\sigma}$.) These disk pieces induce a path in $H$ in between the two pieces from $d_a$, and their cost therefore upper bounds the cost of $D$. We may bound the cost of these
disk pieces by the number of edges on the boundary of $f$ (with respect to $B^* \cup \overline{\sigma}$). The number of such edges is $O(|B^*|+|\overline{\sigma}|)=O(|B^*|)$.

\begin{theorem}
\label{thm:two-point}
Let $s$ and $t$ be two points in the plane, and $G$  a set of
disks such that no disk in $G$ contains either $s$ or $t$, but $G$ separates $s$ and $t$.
There is a polynomial time algorithm that takes such  $G$, $s$, and $t$ as input, and outputs
a subset $B \subseteq G$ that separates $s$ and $t$; the size of $B$ is at most 
a multiplicative constant of the size of the smallest subset $B^* \subseteq G$ that 
separates $s$ and $t$.
\end{theorem}

\section{Separating Multiple Points}
\label{sec:improved}
We now present a polynomial time algorithm that yields an $O(1)$ approximation to the problem of finding a minimum
subset of $\inD$ that separates the set $P$ of points. The algorithm is obtained by calling $\recSep(P)$, where
$\recSep(Q)$, for any $Q \subseteq P$ is the following recursive procedure:

\begin{enumerate}
\item If $|Q| \leq 1$, return $\emptyset$.

\item For every pair of points $s,t \in Q$, invoke the algorithm
of Theorem \ref{thm:two-point} (with $G \leftarrow \inD$) to 
find a subset $B_{s,t} \subseteq \inD$ such that $B_{s,t}$ separates $s$ and $t$.

\item Let $B$ denote the minimum size subset $B_{s,t}$ over all pairs $s$ and $t$ considered.

\item Consider the partition of $Q$ into subsets so that each subset corresponds to points
in the same face (with respect to $B$). Suppose $Q_1, \ldots, Q_{\tau}$ are the subsets in
this partition. Note that $\tau \geq 2$, since $B$ separates some pair of points in $Q$.

\item Return $B \cup \bigcup_{j=1}^{\tau} \recSep(Q_j)$.
\end{enumerate}

Clearly, $\recSep(P)$ yields a separator for $P$. To bound the size of this separator, let us define a set $\Q$ that contains as its
element any $Q \subseteq P$ such that $|Q| \geq 2$ and $\recSep(Q)$ is called somewhere within the call to $\recSep(P)$. For any
$Q \in \Q$, define $B_Q$ to be the set $B$ that is computed in the body of the call to $\recSep(Q)$. Notice that
$\recSep(P)$ returns $\cup_{Q \in \Q} B_Q$.

Now we ``charge'' each such $B_Q$ to an arbitrary  point within $p_Q \in Q$ in such a way that no point in $P$ is charged more than once.
A moment's thought reveals that this is indeed possible. (In a tree where each interval node has degree at least $2$, the number of
leaves is greater than the number of internal nodes.) 

Let $\opt$ denote the optimal separator for $P$ and let $F_Q \subseteq \opt$ denote the disks that contribute to the boundary of
the face (in the arrangement of $\opt$) containing $p_Q$. We claim that $|B_Q| = O(|F_Q|)$; indeed $F_Q$ separates $p_Q \in Q$ from
any point in $P$, and thus any point in $Q$. Thus for any $t \in Q \setminus \{p_Q\}$, we have $|B_Q| \leq B_{p_Q,t} = O(|F_Q|)$. 

We thus have 
\[ \bigcup_{Q \in \Q} |B_Q| \leq \sum_{Q \in \Q} O(|F_Q|) = O(|\opt|), \]
where the last equality follows from union complexity.

We have derived the main result of this paper:

\begin{theorem}
\label{thm: main}
Let $\inD$ be a set of $n$ unit disks and $P$ a set of $k$ points such that $\inD$ separates $P$. There is a polynomial time algorithm that takes as input such $\inD$ and $P$, and returns a subset 
$\outD \subseteq \inD$ of disks that also separates $P$, with the guarantee that  $|\outD|$ is within
a multiplicative $O(1)$ of the smallest subset of $\inD$ that separates $P$.
\end{theorem}

\section{Conclusions}
\label{sec:conclusions}

We have a presented an $O(1)$-approximation algorithm for finding the minimum subset of a given set
$\inD$ of disks that separates a given set of points $P$. One way to understand our contribution is
as follows. Suppose we had at our disposal an efficient algorithm that optimally separates a single
point $p \in P$ from every other point in $P$. Then applying this algorithm for each point in $P$,
we get a separator for $P$. That the size of this separator is within $O(1)$ of the optimal is
an easy consequence of union complexity. However, we only have at our disposal an efficient algorithm
for a weaker task: that of approximately separating two given points in $P$. What we have shown 
is that even this suffices for the task of obtaining an $O(1)$ approximation to the overall problem.

It is easy to see that our algorithm and the approximation guarantee generalize, for example, to
the case when the disks have arbitrary and different radii.

\paragraph{Acknowledgements.} We thank Alon Efrat for discussions that led to the formulation of the problem, and Sariel Har-Peled for discussions that led to the algorithm described here. 
\bibliographystyle{plain}
\bibliography{barcov}

\appendix
\end{document}

%% file: separate.pstex_t
\begin{picture}(0,0)%
\includegraphics{separate.pstex}%
\end{picture}%
\setlength{\unitlength}{1973sp}%
\begingroup\makeatletter\ifx\SetFigFont\undefined%
\gdef\SetFigFont#1#2#3#4#5{%
  \reset@font\fontsize{#1}{#2pt}%
  \fontfamily{#3}\fontseries{#4}\fontshape{#5}%
  \selectfont}%
\fi\endgroup%
\begin{picture}(4824,3024)(2389,-4573)
\end{picture}%

%% file: separate2.pstex_t
\begin{picture}(0,0)%
\includegraphics{separate2.pstex}%
\end{picture}%
\setlength{\unitlength}{1973sp}%
\begingroup\makeatletter\ifx\SetFigFont\undefined%
\gdef\SetFigFont#1#2#3#4#5{%
  \reset@font\fontsize{#1}{#2pt}%
  \fontfamily{#3}\fontseries{#4}\fontshape{#5}%
  \selectfont}%
\fi\endgroup%
\begin{picture}(4824,3024)(2389,-4573)
\end{picture}%

%% file: fig111.pstex_t
\begin{picture}(0,0)%
\includegraphics{fig111.pstex}%
\end{picture}%
\setlength{\unitlength}{987sp}%
\begingroup\makeatletter\ifx\SetFigFont\undefined%
\gdef\SetFigFont#1#2#3#4#5{%
  \reset@font\fontsize{#1}{#2pt}%
  \fontfamily{#3}\fontseries{#4}\fontshape{#5}%
  \selectfont}%
\fi\endgroup%
\begin{picture}(9798,6318)(802,-6355)
\put(8101,-2911){\makebox(0,0)[lb]{\smash{{\SetFigFont{12}{14.4}{\rmdefault}{\mddefault}{\updefault}{\color[rgb]{0,0,0}t}%
}}}}
\put(3901,-5986){\makebox(0,0)[lb]{\smash{{\SetFigFont{12}{14.4}{\rmdefault}{\mddefault}{\updefault}{\color[rgb]{0,0,0}$f_s$}%
}}}}
\put(3226,-3061){\makebox(0,0)[lb]{\smash{{\SetFigFont{12}{14.4}{\rmdefault}{\mddefault}{\updefault}{\color[rgb]{0,0,0}s}%
}}}}
\put(7576,-6286){\makebox(0,0)[lb]{\smash{{\SetFigFont{12}{14.4}{\rmdefault}{\mddefault}{\updefault}{\color[rgb]{0,0,0}$f_t$}%
}}}}
\end{picture}%

%% file: fig112.pstex_t
\begin{picture}(0,0)%
\includegraphics{fig112.pstex}%
\end{picture}%
\setlength{\unitlength}{987sp}%
\begingroup\makeatletter\ifx\SetFigFont\undefined%
\gdef\SetFigFont#1#2#3#4#5{%
  \reset@font\fontsize{#1}{#2pt}%
  \fontfamily{#3}\fontseries{#4}\fontshape{#5}%
  \selectfont}%
\fi\endgroup%
\begin{picture}(9798,6393)(802,-6430)
\put(8101,-2911){\makebox(0,0)[lb]{\smash{{\SetFigFont{12}{14.4}{\rmdefault}{\mddefault}{\updefault}{\color[rgb]{0,0,0}t}%
}}}}
\put(3901,-5986){\makebox(0,0)[lb]{\smash{{\SetFigFont{12}{14.4}{\rmdefault}{\mddefault}{\updefault}{\color[rgb]{0,0,0}$f_s$}%
}}}}
\put(5251,-2461){\makebox(0,0)[lb]{\smash{{\SetFigFont{12}{14.4}{\rmdefault}{\mddefault}{\updefault}{\color[rgb]{0,0,0}$s'$}%
}}}}
\put(7726,-2536){\makebox(0,0)[lb]{\smash{{\SetFigFont{12}{14.4}{\rmdefault}{\mddefault}{\updefault}{\color[rgb]{0,0,0}$t'$}%
}}}}
\put(3226,-3061){\makebox(0,0)[lb]{\smash{{\SetFigFont{12}{14.4}{\rmdefault}{\mddefault}{\updefault}{\color[rgb]{0,0,0}s}%
}}}}
\put(7501,-6361){\makebox(0,0)[lb]{\smash{{\SetFigFont{12}{14.4}{\rmdefault}{\mddefault}{\updefault}{\color[rgb]{0,0,0}$f_t$}%
}}}}
\end{picture}%

%% file: fig11.pstex_t
\begin{picture}(0,0)%
\includegraphics{fig11.pstex}%
\end{picture}%
\setlength{\unitlength}{1460sp}%
\begingroup\makeatletter\ifx\SetFigFont\undefined%
\gdef\SetFigFont#1#2#3#4#5{%
  \reset@font\fontsize{#1}{#2pt}%
  \fontfamily{#3}\fontseries{#4}\fontshape{#5}%
  \selectfont}%
\fi\endgroup%
\begin{picture}(9798,5748)(802,-5785)
\put(5176,-2536){\makebox(0,0)[lb]{\smash{{\SetFigFont{12}{14.4}{\rmdefault}{\mddefault}{\updefault}{\color[rgb]{0,0,0}s'}%
}}}}
\put(8101,-2911){\makebox(0,0)[lb]{\smash{{\SetFigFont{12}{14.4}{\rmdefault}{\mddefault}{\updefault}{\color[rgb]{0,0,0}t}%
}}}}
\put(7726,-2536){\makebox(0,0)[lb]{\smash{{\SetFigFont{12}{14.4}{\rmdefault}{\mddefault}{\updefault}{\color[rgb]{0,0,0}t'}%
}}}}
\put(3301,-2986){\makebox(0,0)[lb]{\smash{{\SetFigFont{12}{14.4}{\rmdefault}{\mddefault}{\updefault}{\color[rgb]{0,0,0}s}%
}}}}
\end{picture}%

%% file: fignew42.pstex_t
\begin{picture}(0,0)%
\includegraphics{fignew42.pstex}%
\end{picture}%
\setlength{\unitlength}{868sp}%
\begingroup\makeatletter\ifx\SetFigFont\undefined%
\gdef\SetFigFont#1#2#3#4#5{%
  \reset@font\fontsize{#1}{#2pt}%
  \fontfamily{#3}\fontseries{#4}\fontshape{#5}%
  \selectfont}%
\fi\endgroup%
\begin{picture}(9314,10151)(738,-8924)
\put(5926,-1486){\makebox(0,0)[lb]{\smash{{\SetFigFont{12}{14.4}{\rmdefault}{\mddefault}{\updefault}{\color[rgb]{0,0,0}$a$}%
}}}}
\put(9976,-3211){\makebox(0,0)[lb]{\smash{{\SetFigFont{12}{14.4}{\rmdefault}{\mddefault}{\updefault}{\color[rgb]{0,0,0}t}%
}}}}
\put(3976,-3286){\makebox(0,0)[lb]{\smash{{\SetFigFont{12}{14.4}{\rmdefault}{\mddefault}{\updefault}{\color[rgb]{0,0,0}$f^*$}%
}}}}
\put(5176,-7411){\makebox(0,0)[lb]{\smash{{\SetFigFont{12}{14.4}{\rmdefault}{\mddefault}{\updefault}{\color[rgb]{0,0,0}$b$}%
}}}}
\put(3826,-1636){\makebox(0,0)[lb]{\smash{{\SetFigFont{12}{14.4}{\rmdefault}{\mddefault}{\updefault}{\color[rgb]{0,0,0}s}%
}}}}
\end{picture}%

%% file: fignew421.pstex_t
\begin{picture}(0,0)%
\includegraphics{fignew421.pstex}%
\end{picture}%
\setlength{\unitlength}{868sp}%
\begingroup\makeatletter\ifx\SetFigFont\undefined%
\gdef\SetFigFont#1#2#3#4#5{%
  \reset@font\fontsize{#1}{#2pt}%
  \fontfamily{#3}\fontseries{#4}\fontshape{#5}%
  \selectfont}%
\fi\endgroup%
\begin{picture}(9239,10151)(738,-8924)
\put(5926,-1411){\makebox(0,0)[lb]{\smash{{\SetFigFont{12}{14.4}{\rmdefault}{\mddefault}{\updefault}{\color[rgb]{0,0,0}$a$}%
}}}}
\put(3526,-3136){\makebox(0,0)[lb]{\smash{{\SetFigFont{12}{14.4}{\rmdefault}{\mddefault}{\updefault}{\color[rgb]{0,0,0}$f$}%
}}}}
\put(9901,-2761){\makebox(0,0)[lb]{\smash{{\SetFigFont{12}{14.4}{\rmdefault}{\mddefault}{\updefault}{\color[rgb]{0,0,0}t}%
}}}}
\put(5176,-7486){\makebox(0,0)[lb]{\smash{{\SetFigFont{12}{14.4}{\rmdefault}{\mddefault}{\updefault}{\color[rgb]{0,0,0}$b$}%
}}}}
\put(3676,-1636){\makebox(0,0)[lb]{\smash{{\SetFigFont{12}{14.4}{\rmdefault}{\mddefault}{\updefault}{\color[rgb]{0,0,0}s}%
}}}}
\put(5626,-661){\makebox(0,0)[lb]{\smash{{\SetFigFont{12}{14.4}{\rmdefault}{\mddefault}{\updefault}{\color[rgb]{0,0,0}$d_a$}%
}}}}
\put(7276,-6736){\makebox(0,0)[lb]{\smash{{\SetFigFont{12}{14.4}{\rmdefault}{\mddefault}{\updefault}{\color[rgb]{0,0,0}$d_b$}%
}}}}
\end{picture}%